\documentclass[journal,twoside,web]{ieeecolor}
\usepackage{amsmath}
\usepackage{amsfonts}
\usepackage{amssymb}
\usepackage{graphicx}
\usepackage{lcsys}
\usepackage{tablefootnote}
\usepackage{stackrel}
\usepackage{microtype}
\usepackage{dsfont}

\newtheorem{theorem}{Theorem}[section]
\newtheorem{lemma}[theorem]{lemma}
\newtheorem{proposition}[theorem]{Proposition}

\newtheorem{remark}[theorem]{Remark}
\newtheorem{assumption}[theorem]{Assumption}
\begin{document}
\title{\vspace{0.25in}
Collaborative Indirect Influencing and Control on Graphs using Graph
Neural Networks}
\author{Max L. Gardenswartz, Brandon C. Fallin, Cristian F. Nino, and Warren
E. Dixon \thanks{Max L. Gardenswartz, Brandon C. Fallin, Cristian F. Nino, and Warren
E. Dixon are with the Department of Mechanical and Aerospace Engineering,
University of Florida, Gainesville, FL 32611, USA. Email: \{mgardenswartz,
brandonfallin, cristian1928, wdixon\}@ufl.edu.}\thanks{This research is supported in part by AFOSR grant FA8651-24-1-0018.
Any opinions, findings, and conclusions or recommendations expressed
in this material are those of the author(s) and do not necessarily
reflect the views of the sponsoring agency.}}
\maketitle
\thispagestyle{empty}
\begin{abstract}
This paper presents a novel approach to solving the indirect influence
problemma in networked systems, in which cooperative nodes must regulate
a target node with uncertain dynamics to follow a desired trajectory.
We leverage the message-passing structure of a graph neural network
(GNN), allowing nodes to collectively learn the unknown target dynamics
in real time. We develop a novel GNN-based backstepping control strategy
with formal stability guarantees derived from a Lyapunov-based analysis.
Numerical simulations are included to demonstrate the performance
of the developed controller.
\end{abstract}


\section{Introduction}

Networks of interconnected agents operate across a variety of domains
including social media, financial markets, traffic systems, and crowd
management. The flow of information and actions in these networks
is formalized through graph theory, where individual actors are represented
as nodes and connections as edges. This representation of a system
enables analysis of how individual nodes' actions propagate through
the network. In scenarios like guiding wildlife populations or regulating
opinion dynamics in virtual networks, individual nodes must leverage
inter-agent interactions to accomplish network-wide goals \cite{King.Portugal.ea2023,Bazzan.Klugl2009}.
We refer to these scenarios as the indirect influence problem.

Various methodologies have been employed to address the indirect influence
problem. Model-based implicit control strategies have enabled coordinated
influence of uncooperative targets with heterogeneous nonlinear dynamics
\cite{Sebastian.Montijano.ea2022} and \emph{switched} systems analysis
has facilitated the development of indirect influence strategies based
on system characteristics of each uncooperative target \cite{Licitra.Bell.ea2019}.
Additionally, containment control strategies have permitted influencing
agents to guide follower agents away from obstacles \cite{Cao2009}.

Existing approaches to the indirect influence problem face a variety
of limitations. Model-based strategies struggle to meet control objectives
with uncertain or unstructured dynamics. Methods relying on individual
leaders face scalability challenges when confronted with numerous
targets, or when targets exhibit high mobility relative to the leaders.
Further, when influence must be coordinated across multiple actors
to regulate a target's state, individual learning fails to capitalize
on shared knowledge. Based on these drawbacks, learning-based methods
with distributed information sharing are motivated for addressing
the indirect influence problem.

Neural networks are powerful tools for approximating continuous functions
over compact domains \cite{Stone1948,Hornik1991,Kidger.Lyons2020}.
Recent work has extended Lyapunov-based adaptive control methods from
shallow neural networks to DNNs, which have been shown to provide
improved function approximation performance with fewer parameters
\cite{Goodfellow.Bengio.ea2016,Rolnick.Tegmark2018,LeCun2015}. Notably,
in \cite{Nino.Patil.ea2023}, a single leader agent uses a Lyapunov-based
DNN (Lb-DNN) to learn unknown inter-agent interaction dynamics in
real time, allowing for the leader to regulate each uncontrolled agent
to within a ball of the desired state. This approach has been extended
to address target tracking with partial state feedback, enabling agents
to reconstruct target states from limited measurements while simultaneously
learning unknown target dynamics, and distributed adaptation approaches
have further enhanced these methods by enabling coordinated learning
across multiple agents \cite{Nino.Patil.ea2024}. Lb-DNNs have also
been leveraged with approximate dynamic programming to explore optimal
online solutions to the indirect influencing problem \cite{Makumi.Bell.ea2024,Philor.Makumi.ea2024}.

Graph neural networks (GNNs) have emerged as a powerful tool for distributed
control applications due to their message-passing framework, which
naturally facilitates information exchange between networked agents.
This framework provides specific advantages for addressing the indirect
influence problem by encoding the network structure directly into
the neural network architecture. Previously, GNNs have been used to
perform autonomous path planning in multi-robot systems \cite{Ji.Li.ea2021}
and develop decentralized controllers for multi-agent flocking \cite{Gama.Tolstaya.ea2021}.
When applied to multi-agent control, GNNs are typically trained offline
and subsequently deployed in a decentralized, open-loop manner \cite{Li.Gama.ea2020,SanchezGonzalez.Heess.ea2018}.
However, the data sets against which these networks are trained may
not sufficiently represent all operating conditions, and the resulting
systems may be unable to adapt online to disturbances. The combination
of GNNs with Lyapunov-based update laws (i.e., Lb-GNNs) addresses
these shortcomings, enabling real-time adaptation for distributed
control tasks \cite{arxivFallin.Nino.ea2025}.

This work addresses an indirect influence problem using a deep Lb-GNN-based
backstepping controller to approximate unknown inter-agent interaction
dynamics. We formulate the problem with a team of cooperative nodes
seeking to influence a target node and regulate it to follow a desired
trajectory. We leverage the deep Lb-GNN's inherent ability to share
learned information, allowing the team of cooperative nodes to collectively
approximate the unknown target dynamics in a distributed manner. Our
method achieves online learning without prior data collection, making
it adaptable to changing system conditions. Our approach scales effectively
with network size by distributing the influence workload across multiple
influencing agents. The analytic results are empirically validated
with a numerical simulation.

\section{Preliminaries}

\subsection{Notation}

Let $\mathbb{R}$ and $\mathbb{Z}$ denote the sets of reals and integers,
respectively. For $x\in\mathbb{R}$, let $\mathbb{R}_{\geq x}\triangleq[x,\infty)$
and $\mathbb{Z}_{\geq x}\triangleq\mathbb{R}_{\geq x}\cap\mathbb{Z}$.
Let ${\rm {\bf h}}_{i}$ represent the $i^{\text{th}}$ standard basis
in $\mathbb{R}^{N}$, where the standard basis is made up of vectors
that have one entry equal to $1$ and the remaining $N-1$ entries
equal to $0$. The $p\times p$ identity matrix and the $p\times1$
column vector of ones are denoted by $I_{p}$ and \textbf{$\boldsymbol{1}_{p}$},
respectively. The $m\times n$ column vector of zeros is denoted by
$\boldsymbol{0}_{m\times n}$. The enumeration operator $[\cdot]$
is defined as $[N]\triangleq\{1,2,\ldots,N\}.$ Given a positive integer
$N$ and collection $\{x_{i}\}{}_{i\in[N]}$, let $[x_{i}]{}_{i\in[N]}\triangleq[x_{1},x_{2},\ldots,x_{N}]\in\mathbb{R}^{m\times N}$
for $x_{i}\in\mathbb{R}^{m}$. The cardinality of a set $A$ is denoted
by $|A|.$ For a set $A$ and an input $x$, the indicator function
is denoted by $\mathds{1}_{A}(x)$, where $\mathds{1}_{A}(x)=1$ if
$x\in A$, and $\mathds{1}_{A}(x)=0$ otherwise. For $x\in\mathbb{R}^{m}$,
let $\text{diag}\left\{ \cdot\right\} $ denote the diagonalization
operator which assigns the $i^{\text{th}}$ value of $x$ to the $ii^{\text{th}}$
entry of an output $m\times m$ matrix for all $i\in[m]$. Similarly,
let $\text{blkdiag}\left\{ \cdot\right\} $ denote the block diagonalization
operator. For $A\in\mathbb{R}^{m\times n}$ and $B\in\mathbb{R}^{p\times q}$,
\[
\text{blkdiag}\left\{ A,B\right\} \triangleq\begin{bmatrix}A & \boldsymbol{0}_{m\times q}\\
\boldsymbol{0}_{p\times n} & B
\end{bmatrix},
\]
where $\boldsymbol{0}_{m\times q}\in\mathbb{R}^{m\times q}$ is the
$m\times q$ matrix of zeros and $\text{blkdiag}\left\{ A,B\right\} \in\mathbb{R}^{(m+p)\times(n+q)}$.
Given $A\in\mathbb{R}^{m\times n}$ with columns $[a_{i}]_{i\in[n]}^{\top}\subset\mathbb{R}^{b}$,
$\text{vec}(A)\triangleq[a_{1}^{\top},a_{2}^{\top},\ldots,a_{n}^{\top}]{}^{\top}\in\mathbb{R}^{mn}$.

\subsection{Graphs}

For $N\in\mathbb{Z}_{>0}$, let $G\triangleq(V,E)$ be a static and
undirected graph with node set $V\triangleq[N]$ and edge set $E\subseteq V\times V$.
The edge $(i,j)\in E$ if and only if the node $i$ can send information
to node $j$. In this work, $G$ is undirected, so $(i,j)\in E$ if
and only if $(j,i)\in E$. Let $A\triangleq[a_{ij}]\in\mathbb{R}^{N\times N}$
denote the adjacency matrix of $G$, where $a_{ij}=\mathds{1}_{E}(i,j)$
and $a_{ii}=0$ for all $i\in V$. An undirected graph is connected
if and only if there exists a sequence of edges in $E$ between any
two nodes in $V$. The neighborhood of node $i$ is denoted by $\mathcal{N}_{i}$,
where $\mathcal{N}_{i}\triangleq\{j\in V:(j,i)\in E\}$. The augmented
neighborhood of node $i$ is denoted by $\overline{\mathcal{N}}_{i}$,
where $\overline{\mathcal{N}}_{i}\triangleq\mathcal{N}_{i}\cup\{i\}$.
The set of node $i$'s $k$-hop neighbors is denoted by $\mathcal{N}_{i}^{k}$.
The degree matrix of $G$ is denoted by $D\triangleq\text{diag}(A\boldsymbol{1}_{N})$.
The Laplacian matrix of $G$ is denoted by $\mathcal{L}_{G}\triangleq D-A\in\mathbb{R}^{N\times N}$.
In this work, we focus on a class of communication networks where
information flows through all nodes, which occurs when the graph $G$
is connected.

The set of permutations on $[N]$ is denoted by $S_{N}$. For a graph
$G$ and a permutation $S_{N}$, we define the graph permutation operation
as $p\ast G$, where $p\in S_{N}$. Formally, two graphs $G_{1}$
and $G_{2}$ are isomorphic if they have the same number of nodes
and there exists a permutation such that $G_{1}=p\ast G_{2}$ \cite{Azizian.Lelarge2020}.

\subsection{Graph Neural Networks\label{subsec:Graph-Neural-Networks}}

The GNN architecture employs a message-passing framework in which
nodes exchange and update vector-valued messages using feedforward
layers. These messages are the outputs of each GNN layer calculated
at the node level. The GNN processes an input graph $G$ with a set
of node features to generate node embeddings, which are the outputs
of the final GNN layer at each node. In this framework, each node's
output in a layer is informed by its previous layer output and messages
received from neighboring nodes. We use superscripts to denote layer-specific
elements. For example, $W_{i}^{(k)}$ denotes the $i^{\text{th}}$
node's weights for the $k^{\text{th}}$ GNN layer. Subscripts indicate
the specific node where an embedding or function is applied.

Let the activation function for the $k^{\text{th}}$ GNN layer at
node $i$ be denoted by $\sigma^{(k)}(\cdot)$, where $\sigma(\cdot)$
is an element-wise, smooth, bounded nonlinearity with bounded first
and second derivatives appended with $1$ to allow for biases such
that for $x\in\mathbb{R}^{m}$, $\sigma(x)=[\sigma(x_{0}),\ldots,\sigma(x_{m}),1]{}^{\top}$.
Let the output of layer $j$ for node $i$ of the GNN be denoted by
$\phi_{i}^{(j)}$. Let the aggregation function for the GNN at node
$i$ be denoted by $\sum_{m\in\mathcal{\overline{N}}_{i}}\phi_{m}^{(k-1)}$.
Additionally, let $d^{(j)}$ represent the number of features at the
$j^{\text{th}}$ layer of the GNN for $j=0,\ldots,k$, where $j$
denotes the layer index. Let $d^{(in)}$ denote the dimension of the
GNN input at the base layer of each node. Let $d^{(out)}$ denote
the output dimension of the output layer of each node. Let $d^{(j)}=d^{(in)}+1$
when $j=-1$ and $d^{(j)}=d^{(out)}$ when $j=k$. 

Next, we define the deep GNN architecture for an arbitrary number
of layers. Let $\overline{\kappa}_{i}$ denote the $i^{\text{th}}$
node's input augmented with a bias term such that $\overline{\kappa}_{i}\triangleq[\kappa_{i}^{\top},1]{}^{\top}\in\mathbb{R}^{d^{(in)}+1}$.
Let $\overline{\boldsymbol{\kappa}}\triangleq[\overline{\kappa}_{1},\ldots,\overline{\kappa}_{N}]=[\overline{\kappa}_{i}]{}_{i\in V}\in\mathbb{R}^{(d^{(in)}+1)\times N}$.
To distinguish between node indices in cascading layers, we let $m^{(j)}$
denote an arbitrary index corresponding to the $j^{\text{th}}$ layer,
where $m^{(j)}\in V$. Then, we let $\boldsymbol{{\rm m}}^{(j)}$
denote $m^{(j)}\in V$. Let $\boldsymbol{\phi}^{(j)}\triangleq[\phi_{1}^{(j)},\ldots,\phi_{N}^{(j)}]=[\phi_{i}^{(j)}]_{m^{(j)}}\in\mathbb{R}^{(d^{(j-1)}+1)\times N}$.
Let $\overline{A}\in\mathbb{R}^{N\times N}$ represent the adjacency
matrix with self loops, and $\overline{A}_{i}\in\mathbb{R}^{1\times N}$
represent the $i^{\text{th}}$ row of the adjacency matrix. The GNN
architecture at node $i$ is expressed as
\begin{equation}
\phi_{i}^{(j)}\triangleq\begin{cases}
\sigma^{(j)}\left(W_{i}^{(j)\top}\overline{\boldsymbol{\kappa}}\overline{A}_{i}^{\top}\right), & j=0,\\
\sigma^{(j)}\left(W_{i}^{(j)\top}\boldsymbol{\phi}^{(j-1)}\overline{A}_{i}^{\top}\right), & j=1,\ldots,k-1,\\
W_{i}^{(j)\top}\phi_{i}^{(j-1)}, & j=k,
\end{cases}\label{eq:deep-gnn-architecture}
\end{equation}
for all $i\in V$, an input $y^{(j)}\in\mathbb{R}^{m}$, the partial
derivative $\frac{\partial\sigma^{(j)}}{\partial y^{(j)}}:\mathbb{R}^{d^{(j)}}\to\mathbb{R}^{(d^{(j)}+1)\times d^{(j)}}$
of the activation function vector at the $j^{\text{th}}$ layer with
respect to its input is given as $[\sigma^{\prime(j)}(y_{1}){\rm {\bf h}_{1}},\ldots,\sigma^{\prime(j)}(y_{d^{(j)}}){\rm {\bf h}}_{d^{(j)}},\boldsymbol{0}_{d^{(j)}}]{}^{\top}\in\mathbb{R}^{(d^{(j)}+1)\times d^{(j)}}$,
where ${\rm {\bf h}}_{i}$ is the $i^{\text{th}}$ standard basis
in $\mathbb{R}^{d^{(j)}}$ and $\boldsymbol{0}_{d^{(j)}}$ is the
zero vector in $\mathbb{R}^{d^{(j)}}$. The vector of weights for
the GNN at node $i$ is defined as $\theta_{i}\in\mathbb{R}^{p_{{\rm GNN}}}$,
where
\begin{equation}
\theta_{i}\triangleq\begin{bmatrix}\text{vec}\left(W_{i}^{(0)}\right)^{\top}, & \ldots, & \text{vec}\left(W_{i}^{(k)}\right)^{\top}\end{bmatrix}^{\top},\label{eq:deep-gnn-weights}
\end{equation}
and $p_{{\rm GNN}}\triangleq\sum_{j=0}^{k}(d^{(j)})(d^{(j-1)}+1)$.
For each node $i\in V$, the first partial derivative of the deep
GNN architecture at node $i$ in (\ref{eq:deep-gnn-architecture})
with respect to (\ref{eq:deep-gnn-weights}) is
\[
\frac{\partial\phi_{i}}{\partial\theta_{i}}\triangleq\begin{bmatrix}\frac{\partial\phi_{i}}{\partial\text{vec}\left(W_{i}^{(0)}\right)}, & \ldots, & \frac{\partial\phi_{i}}{\partial\text{vec}\left(W_{i}^{(k)}\right)}\end{bmatrix},
\]
where $\frac{\partial\phi_{i}}{\partial\theta_{i}}\in\mathbb{R}^{d^{(out)}\times p_{{\rm GNN}}}$\cite[Lemma 1]{arxivFallin.Nino.ea2025}.
The closed form of the partial derivatives of $\phi_{i}$ with respect
to $\text{vec}(W_{i}^{(\ell)})$ for all $\ell=0,\ldots,k$ is defined
in \cite[Table 1]{arxivFallin.Nino.ea2025}. 

Let $\mathcal{G}_{N}$ denote the set of undirected graphs with $N$
nodes. Let $F\in\mathbb{R}^{Nm}$ denote the ensemble feature vector
of the graph, where each node has features in $\mathbb{R}^{m}.$ A
function $f:\mathcal{G}_{N}\times\mathbb{R}^{Nm}\to\mathbb{R}^{N\ell}$
on a graph $G$ is called equivariant if $f(p\ast G,p\ast F)=p\ast f(G,F)$
for every permutation $p\in S_{N}$, every $G\in\mathcal{G}_{N}$,
and every $F\in\mathbb{R}^{Nm}$ \cite{Azizian.Lelarge2020}. An example
of an equivariant function on a graph is a continuous function of
a node's features evaluated at that node. In this work, we approximate
equivariant functions of the graph using a deep Lb-GNN.

The expressive power of deep GNNs is intrinsically linked to the message-passing
framework \cite{Hamilton2020}. This expressive ability can be formalized
through the 1-Weisfeiler-Leman (1-WL) test, a graph isomorphism heuristic
that iteratively refines node labels based on neighborhood information.
The 1-WL test's ability to distinguish graph structures has been shown
to be representative of the discriminative ability of the GNN architecture.
The message-passing GNN in (\ref{eq:deep-gnn-architecture}) universally
approximates equivariant, graph-valued functions that are less separating
the 2-WL test \cite{Azizian.Lelarge2020}.

\section{Problem Formulation\label{sec:Problem-Formulation}}

In this section, we consider a network of $N$ nodes tasked with collaboratively
influencing a target node to follow a desired trajectory. The target
node has unknown, unstructured interaction dynamics with respect to
each influencing node and unknown, unstructured drift dynamics. The
influencing nodes have homogeneous unknown, unstructured dynamics.
The deep Lb-GNN architecture is used to learn the target node's unknown
drift dynamics, the target node's unknown interaction dynamics, and
the unknown interaction dynamics between influencing nodes.

\subsection{Systems Dynamics}

Consider a network of $N$ influencing nodes indexed by $i\in V$
with the static, undirected communication graph $G=(V,E)$, where
$N\geq2$. The influencing nodes' objective is to collaboratively
regulate the state of a target node to follow a desired trajectory.
All influencing nodes can measure their relative states with respect
to the target node. The state of the $i^{\text{th}}$ influencing
node is denoted by $y_{i}\in\mathbb{R}^{n}$, for all $i\in V$. Let
$Q_{i}\triangleq[(y_{m}\cdot\mathds{1}_{\overline{\mathcal{N}}_{i}}(m))^{\top}]_{m\in V}^{\top}\in\mathbb{R}^{nN}$
denote the concatenated states of the $i^{\text{th}}$ influencing
node and its neighbors. The dynamics of the $i^{\text{th}}$ influencing
node are 
\begin{equation}
\dot{y}_{i}=f\left(Q_{i}\right)+u_{i},\label{eq:influencer-dynamics}
\end{equation}
where $u_{i}\in\mathbb{\mathbb{R}}^{n}$ denotes the $i^{\text{th}}$
influencing node's control input and $f:\mathbb{\mathbb{R}}^{nN}\to\mathbb{\mathbb{R}}^{n}$
is an unknown, unstructured function that denotes the interaction
dynamics between node $i$ and its neighboring nodes, indexed by $j\in\mathcal{N}_{i}$.
The target node is indirectly influenced by an unknown interaction
function denoted by $g:\mathbb{R}^{n}\times\mathbb{R}^{n}\to\mathbb{R}$,
which depends on the target node's state and an individual influencing
node's state. Let the target node's state be denoted by $x_{0}\in\mathbb{R}^{n}.$
The dynamics of the target node are given as
\begin{equation}
\dot{x}_{0}=h\left(x_{0}\right)+\sum_{i\in V}g\left(x_{0},y_{i}\right)\left(x_{0}-y_{i}\right),\label{eq:target-dynamics}
\end{equation}
where $h:\mathbb{R}^{n}\to\mathbb{R}^{n}$ is an unknown, unstructured
function that denotes the target node's drift dynamics. The functions
$f$, $g$, and $h$ are assumed to be continuous. In practice, the
target node's dynamics are often constrained by initial conditions
and environmental factors which limit its state variations. This is
formalized in the following assumption.
\begin{assumption}
\label{ass:interaction-function-bounds}\cite[Assumption 1]{Le.Luo.ea2020}
There exist known constants $\overline{g}\in\mathbb{R}_{>0}$ and
$\overline{h}\in\mathbb{R}_{>0}$ such that $\|g(x_{0},y_{i})\|\leq\overline{g}$
and $\|h(x_{0})\|\leq\overline{h}$.
\end{assumption}

\subsection{Control Objective}

The influencing nodes' control objective is to regulate $x_{0}$ to
a user-defined desired trajectory $x_{d}\in\mathbb{R}^{n}$ of class
$\mathcal{C}^{1}$, despite the target node's uncertain drift and
interaction dynamics described by (\ref{eq:target-dynamics}). To
this end, we define the state tracking error $e\in\mathbb{R}^{n}$
as
\begin{equation}
e\triangleq x_{0}-x_{d}.\label{eq:tracking-error}
\end{equation}

\begin{assumption}
\label{ass:goal-location-bounds}\cite[Assumption 1]{Nino.Patil.ea2024}
There exist known constants $\overline{x}_{d}\in\mathbb{R}$ and $\overline{\dot{x}}_{d}\in\mathbb{R}$
such that $\|x_{d}(t)\|\leq\overline{x}_{d}$ and $\|\dot{x}_{d}(t)\|\leq\overline{\dot{x}}_{d}$
for all $t\in[t_{0},\infty)$.
\end{assumption}
Since the target node's dynamics in (\ref{eq:target-dynamics}) do
not explicitly contain a control input, a backstepping control technique
is employed. The signal $y_{i}$ is treated as a virtual control input.
We define a backstepping error $\eta_{i}\in\mathbb{R}^{n}$ for the
$i^{\text{th}}$ influencing node as
\begin{equation}
\eta_{i}\triangleq y_{d,i}-y_{i},\label{eq:backstepping-error}
\end{equation}
where $y_{d,i}\in\mathbb{R}^{n}$ denotes the $i^{\text{th}}$ influencing
node's desired trajectory for all $i\in V$. Substituting (\ref{eq:target-dynamics})
into the time derivative of (\ref{eq:tracking-error}) and using (\ref{eq:backstepping-error})
yields
\begin{equation}
\dot{e}=h\left(x_{0}\right)-\dot{x}_{d}+\sum_{i\in V}g\left(x_{0},y_{i}\right)\left(x_{0}+\eta_{i}-y_{d,i}\right).\label{eq:error-derivative-1}
\end{equation}

\section{Control Synthesis}

Based on the subsequent stability analysis, the desired trajectory
of the $i^{\text{th}}$ influencing node is designed as

\begin{equation}
y_{d,i}=k_{1}e+x_{d},\label{eq:influence-strategy}
\end{equation}
for all $i\in V$, where $k_{1}\in\mathbb{R}_{>0}$ is a user-defined
gain. Substituting (\ref{eq:influence-strategy}) into (\ref{eq:error-derivative-1})
yields
\begin{align}
\dot{e} & =h\left(x_{0}\right)-\dot{x}_{d}+\sum_{i\in V}g\left(x_{0},y_{i}\right)\left(\eta_{i}+\left(1-k_{1}\right)e\right).\label{eq:error-derivative-2}
\end{align}
Substituting (\ref{eq:influencer-dynamics}), (\ref{eq:error-derivative-2}),
and the time derivative of (\ref{eq:influence-strategy}) into the
time derivative of (\ref{eq:backstepping-error}) yields
\begin{equation}
\begin{aligned}\dot{\eta}_{i} & =k_{1}\sum_{j\in V\backslash\mathcal{\overline{N}}_{i}}g\left(x_{0},y_{j}\right)\left(\eta_{j}+\left(1-k_{1}\right)e\right)\\
 & +\left(1-k_{1}\right)\dot{x}_{d}-u_{i}+F\left(R_{i}\right),
\end{aligned}
\label{eq:backstepping-derivative-2}
\end{equation}
where $F\left(R_{i}\right)\triangleq k_{1}h\left(x_{0}\right)-f\left(Q_{i}\right)+\sum_{j\in\mathcal{\overline{N}}_{i}}k_{1}g\left(x_{0},y_{j}\right)\left(\eta_{j}+\left(1-k_{1}\right)e\right)$,
$R_{i}\triangleq[x_{0}^{\top},Q_{i}^{\top}]^{\top}\in\mathbb{R}^{d^{(in)}}$
for all $i\in V$, and $d^{(in)}=n(N+1)$.

\subsection{Function Approximation Capabilities of GNNs}

To approximate the unknown inter-agent dynamics, we leverage the universal
function approximation properties of GNNs established in \cite{Azizian.Lelarge2020}.
To this end, we define the deep Lb-GNN $\Phi\triangleq[\phi_{i}^{\top}]_{i\in V}^{\top}$,
with $p$ total weights at each node $i$. The output of a $k$-layer
Lb-GNN $\Phi$ at node $i$ is a function of its $(k-1)$-hop neighbor's
weights, $(\theta_{j})_{j\in\overline{\mathcal{N}}_{i}^{k-1}}$, and
its $k$-hop neighbor's inputs, $(R_{j})_{j\in\overline{\mathcal{N}}_{i}^{k}}$.
Node $i$'s component of the Lb-GNN $\Phi$ is denoted by $\phi_{i}\triangleq\phi_{i}(R_{i},(R_{j})_{j\in\mathcal{N}_{i}^{k}},\theta_{i},(\theta_{j})_{j\in\mathcal{N}_{i}^{k-1}})$.
Let the input to the Lb-GNN $\Phi$ be denoted as $R\triangleq[R_{i}^{\top}]_{i\in V}^{\top}\in\mathcal{Y}$,
where $\mathcal{Y}\subset\mathbb{R}^{Nd^{(in)}}$. Let $\Omega$ denote
the compact set over which the universal function approximation property
of GNNs holds. We use the Lb-GNN $\Phi$ to approximate the function
$H:\mathbb{R}^{Nd^{(in)}}\to\mathbb{R}^{nN}$, where $H(R)\triangleq[\begin{array}{ccc}
F(R_{1})^{\top}, & \ldots, & F(R_{N})^{\top}\end{array}]^{\top}$, and $F(R_{i})$ denotes the unknown inter-agent interaction dynamics
for all $i\in V$.
\begin{assumption}
\label{ass:less-separating}The function $H$ is less separating than
the 2-WL test.
\end{assumption}
Assumption \ref{ass:less-separating} states that the graph-valued
function $H$ must not reveal more structural information of $G$
than is able to be distinguished through the $2$-WL test. Assumption
\ref{ass:less-separating} is reasonable since, in practice, any continuous
function that at each node aggregates over neighboring nodes and applies
a nonlinearity to this aggregated value is provably less separating
than the 1-WL test \cite{Azizian.Lelarge2020}. The 1-WL test is less
separating than the 2-WL test, so any function of this form can be
universally approximated by a GNN. We refer the reader to \cite{Azizian.Lelarge2020}
for a more complete exposition on the classes of functions that are
able to be universally approximated by GNNs.

If (i) the graph-valued function being approximated by the GNN is
less separating than the 2-WL test and (ii) the input space of the
graph valued function is compact, the universal function approximation
property of GNNs in \cite[Lemma 3]{arxivFallin.Nino.ea2025} holds.
We note that the function $H(R)$ is an equivariant function of the
graph $G$ because the outputs at each node, $F(R_{i})$, are a function
of each node's features and the features of its one-hop neighborhood.
Thus, if the graph $G$ were permuted, $H(R)$ would reflect this
permutation.

Let the ensemble vector of the GNN node weights be denoted as $\theta\triangleq[\theta_{i}^{\top}]_{i\in V}^{\top}\in\mathbb{R}^{pN}$.
Let the loss function for the GNN $\Phi$ be defined as $\mathcal{L}(\theta)\triangleq\intop_{\Omega}(\|H(R)-\Phi(R,\theta)\|^{2}+\sigma\|\theta\|^{2}d\mu(R))$,
where $\mu$ denotes the Lebesgue measure, $\sigma\in\mathbb{R}_{>0}$
denotes a regularizing constant, and $\sigma\|\theta\|^{2}$ represents
$L_{2}$ regularization. Let $\mho\subset\mathbb{R}^{pN}$ denote
a user-selected compact, convex parameter search space with a smooth
boundary, satisfying $\boldsymbol{0}_{pN}\in\text{int}(\mho)$. Additionally,
define $\overline{\theta}\in\mathbb{R}$, where $\overline{\theta}\triangleq\max_{\theta\in\mho}\|\theta\|$.
We denote the ideal parameters of the GNN $\Phi$ as $\theta^{\ast}\in\mho\subset\mathbb{R}^{pN}$
where
\begin{equation}
\theta^{\ast}\triangleq\underset{\theta\in\mho}{\arg\min}\:\mathcal{L}\left(\theta\right).\label{eq:theta_star}
\end{equation}
For clarity in the subsequent analysis, we desire the solutions $\theta^{\ast}$
to (\ref{eq:theta_star}) to be unique. To this end, the following
assumption is made.
\begin{assumption}
\cite[Assumption 3]{arxivFallin.Nino.ea2025} \label{ass:convex-loss-function}
The loss function $\mathcal{L}$ is strictly convex over set $\mho$.
\end{assumption}
\begin{remark}
The universal function approximation property of GNNs was not invoked
in the definition of $\theta^{\ast}.$ The universal function approximation
theorem for GNNs established in \cite[Lemma 3]{arxivFallin.Nino.ea2025}
states that the function space of GNNs is dense in $\mathcal{C}_{E}(\Omega,\mathbb{R}^{nN})$,
which denotes the space of real-valued, continuous, equivariant functions
that map from $\Omega$ to $\mathbb{R}^{nN}$. Let $\varepsilon:\Omega\to\mathbb{R}^{nN}$
denote an unknown function representing the reconstruction error that
is bounded as $\sup_{G\times R\in\Omega}\|H(R)-\Phi(R,\theta)\|\leq\overline{\varepsilon}$.
Therefore, $\intop_{\Omega}\|F(R)-\Phi(R,\theta)\|^{2}d\mu(R)<\varepsilon^{2}\mu(\Omega)$.
As a result, for any prescribed $\overline{\varepsilon}\in\mathbb{R}$,
there exists a GNN $\Phi$ such that for all $i\in V$, there exist
weights $\theta_{i}\in\mathbb{R}^{p}$ which satisfy $\sup_{G\times R\in\Omega}\|F(R)-\Phi(R,\theta)\|\leq\overline{\varepsilon}$.
However, determining a search space $\mho$ for an arbitrary $\varepsilon$
remains an open challenge. So, we allow $\mho$ to be arbitrarily
selected in the above analysis, at the expense of guarantees of the
approximation accuracy.
\end{remark}
The unknown dynamics in (\ref{eq:backstepping-derivative-2}) are
modeled using the GNN $\Phi$ at node $i$ as
\begin{equation}
F\left(R_{i}\right)=\phi_{i}^{\ast}+\varepsilon_{i}(R_{i},(R_{j})_{j\in\mathcal{N}_{i}^{k}}),\label{eq:lb-gnn}
\end{equation}
for all $i\in V$, where $\theta_{i}^{\ast}$ are the ideal parameters
of the GNN $\Phi$ at node $i$, $\phi_{i}^{\ast}\triangleq\phi_{i}(R_{i},(R_{j})_{j\in\mathcal{N}_{i}^{k}},\theta_{i}^{\ast},(\theta_{j}^{\ast})_{j\in\mathcal{N}_{i}^{k-1}})$
and $\varepsilon_{i}\in\mathbb{R}^{n}$ denotes the $i^{\text{th}}$
component of $\varepsilon$ for all $i\in V$, where $\varepsilon\triangleq[\begin{array}{c}
\varepsilon_{i}^{\top}\end{array}]_{i\in V}^{\top}$. The Lb-GNN's approximation of $H(R)$ at node $i$ is given as $\widehat{\phi}_{i}\triangleq\phi_{i}(R_{i},(R_{j})_{j\in\mathcal{N}_{i}^{k}},\widehat{\theta}_{i},(\widehat{\theta}_{j})_{j\in\mathcal{N}_{i}^{k-1}})$,
where $\widehat{\theta}_{i}\in\mathbb{R}^{p}$ denotes the Lb-GNN's
weight estimates at node $i$ that are subsequently designed through
a Lyapunov-based analysis. The approximation objective is to determine
optimal estimates of $\widehat{\theta}$ such that $R\mapsto\Phi(R,\theta)$
approximates $R\mapsto H(R)$ with minimal error for any $R\in\Omega$. 

\subsection{Control Design}

Based on (\ref{eq:backstepping-derivative-2}) and the subsequent
stability analysis, the controller at node $i$ is designed as{\small
\begin{align}
u_{i} & =k_{2}\eta_{i}+\widehat{\phi}_{i}+\sum_{j\in\mathcal{\mathcal{N}}_{i}^{k-1}}\frac{\partial\widehat{\phi}_{i}}{\partial\widehat{\theta}_{j}}\left(\widehat{\theta}_{i}-\widehat{\theta}_{j}\right)+\left(1-k_{1}\right)\dot{x}_{d},\label{eq:controller}
\end{align}
}where $k_{2}\in\mathbb{R}_{>0}$ is a user-defined gain, for all
$i\in V$. Substituting (\ref{eq:lb-gnn}) and (\ref{eq:controller})
into (\ref{eq:backstepping-derivative-2}) yields
\begin{equation}
\begin{aligned}\dot{\eta}_{i} & =k_{1}\sum_{j\in V\backslash\mathcal{\overline{N}}_{i}}g\left(x_{0},y_{j}\right)\left(\eta_{j}+\left(1-k_{1}\right)e\right)\\
 & -k_{2}\eta_{i}+\phi_{i}^{\ast}-\widehat{\phi}_{i}+\varepsilon_{i}-\sum_{j\in\mathcal{\mathcal{N}}_{i}^{k-1}}\frac{\partial\widehat{\phi}_{i}}{\partial\widehat{\theta}_{j}}\left(\widehat{\theta}_{i}-\widehat{\theta}_{j}\right).
\end{aligned}
\label{eq:backstepping-derivative-3}
\end{equation}

\subsection{Adaptive Update Law Design}

Based on the subsequent stability analysis, we design the distributed
adaptive update law for the Lb-GNN as
\begin{align}
\dot{\widehat{\theta}}_{i} & =\text{proj}\left(\aleph_{i}\right),\label{eq:update-law}
\end{align}
where{\small
\begin{align*}
\aleph_{i} & \triangleq\Gamma_{i}\left(\left(\sum_{j\in\mathcal{\overline{N}}_{i}^{k-1}}\frac{\partial\widehat{\phi}_{i}}{\partial\widehat{\theta}_{j}}^{\top}\right)\eta_{i}-k_{3}\left(\sum_{j\in\mathcal{N}_{i}}\left(\widehat{\theta}_{i}-\widehat{\theta}_{j}\right)+\widehat{\theta}_{i}\right)\right),
\end{align*}
}where $k_{3}\in\mathbb{R}_{>0}$ is a user-defined gain, $\Gamma_{i}\in\mathbb{R}^{p\times p}$
is a symmetric, positive-definite, user-defined gain matrix, and the
projection operator ensures $\widehat{\theta}\in\mho$ for all $t\in\mathbb{R}_{\geq0}$,
defined as in \cite[Appendix E.4]{Krstic.Kanellakopoulos.ea1995}.
In the following section, we perform a stability analysis for the
ensemble system.

\section{Stability Analysis}

To quantify the approximation given by (\ref{eq:lb-gnn}), the parameter
estimation error $\widetilde{\theta}_{i}\in\mathbb{R}^{p}$ at node
$i$ is defined as
\begin{equation}
\widetilde{\theta}_{i}\triangleq\theta_{i}^{\ast}-\widehat{\theta}_{i},\label{eq:parameter-estimation-error}
\end{equation}
for all $i\in V$. Applying a first-order Taylor theorem of $\phi_{i}^{\ast}$
about the point $(R_{i},(R_{j})_{j\in\mathcal{N}_{i}^{k}},\widehat{\theta}_{i},(\widehat{\theta}_{j})_{j\in\mathcal{N}_{i}^{k-1}})$
yields
\begin{equation}
\phi_{i}^{\ast}=\widehat{\phi}_{i}+\frac{\partial\widehat{\phi}_{i}}{\partial\widehat{\theta}_{i}}\widetilde{\theta}_{i}+\sum_{j\in\mathcal{\mathcal{N}}_{i}^{k-1}}\frac{\partial\widehat{\phi}_{i}}{\partial\widehat{\theta}_{j}}\widetilde{\theta}_{j}+\sum_{j\in\mathcal{\overline{N}}_{i}^{k-1}}T_{j},\label{eq:taylor-series}
\end{equation}
for all $i\in V$, where $T_{j}\in\mathbb{R}^{n}$ denotes the first
Lagrange remainder, which accounts for the error introduced by truncating
the Taylor approximation after the first-order term. Substituting
(\ref{eq:taylor-series}) into (\ref{eq:backstepping-derivative-3})
and using (\ref{eq:parameter-estimation-error}) yields

\begin{equation}
\begin{aligned}\dot{\eta}_{i} & =k_{1}\sum_{j\in V\backslash\mathcal{\overline{N}}_{i}}g\left(x_{0},y_{j}\right)\left(\eta_{j}+\left(1-k_{1}\right)e\right)\\
 & +\varepsilon_{i}-k_{2}\eta_{i}+\sum_{j\in\mathcal{\mathcal{\overline{N}}}_{i}^{k-1}}\frac{\partial\widehat{\phi}_{i}}{\partial\widehat{\theta}_{j}}\widetilde{\theta}_{i}+\chi_{i},
\end{aligned}
\label{eq:backstepping-derivative-4}
\end{equation}
where $\chi_{i}\triangleq\sum_{j\in\mathcal{\mathcal{N}}_{i}^{k-1}}\frac{\partial\widehat{\phi}_{i}}{\partial\widehat{\theta}_{j}}\left(\theta_{j}^{\ast}-\theta_{i}^{\ast}\right)+\sum_{j\in\mathcal{\overline{N}}_{i}^{k-1}}T_{j}$.
\begin{lemma}
\cite[Lemma 5]{arxivFallin.Nino.ea2025}\label{lemma:lagrange-bounds}
The first Lagrange remainder $T_{i}$ is upper bounded as $\|T_{i}\|\leq\rho_{T_{1},i}(\|\kappa\|)\|\widetilde{\theta}_{i}\|^{2}$
for all $i\in V$, where $\rho_{T_{1},i}:\mathbb{R}_{\geq0}\to\mathbb{R}_{\geq0}$
is a strictly increasing polynomial that is quadratic in the norm
of the ensemble GNN input $\kappa\triangleq[\kappa_{i}^{\top}]_{i\in V}^{\top}\in\mathbb{R}^{Nd^{(\text{in})}}.$
\end{lemma}
Let $\eta\triangleq[\eta_{i}^{\top}]_{i\in V}^{\top}\in\mathbb{R}^{nN}$
denote the ensemble backstepping error vector for all influencing
nodes. Define $\widetilde{\theta}\triangleq[\widetilde{\theta}_{i}^{\top}]_{i\in V}^{\top}\in\mathbb{R}^{pN}$
and $\Gamma\triangleq\text{blkdiag\{\ensuremath{\Gamma_{1}},\ensuremath{\ldots},\ensuremath{\Gamma_{N}}\}}\in\mathbb{R}^{pN\times pN}$.
Define the concatenated state vector as $z\triangleq[\begin{array}{ccc}
e^{\top} & \eta^{\top} & \widetilde{\theta}^{\top}\end{array}]^{\top}\in\mathbb{R}^{\varphi}$, where $\varphi\triangleq n+N(n+p)$. Taking the time derivative
of $z$ and substituting (\ref{eq:error-derivative-2}), (\ref{eq:update-law}),
and (\ref{eq:backstepping-derivative-4}) into the resulting expression
yields{\tiny
\begin{equation}
\dot{z}=\left[\begin{array}{c}
h\left(x_{0}\right)-\dot{x}_{d}+\sum_{i\in V}g\left(x_{0},y_{i}\right)\left(\eta_{i}+\left(1-k_{1}\right)e\right)\\
\left[\left(\begin{array}{c}
k_{1}\sum_{j\in V\backslash\mathcal{\overline{N}}_{i}}g\left(x_{0},y_{j}\right)\left(\eta_{j}+\left(1-k_{1}\right)e\right)+\varepsilon_{i}\\
-k_{2}\eta_{i}+\sum_{j\in\mathcal{\mathcal{\overline{N}}}_{i}^{k-1}}\frac{\partial\widehat{\phi}_{i}}{\partial\widehat{\theta}_{j}}\widetilde{\theta}_{i}+\chi_{i}
\end{array}\right)^{\top}\right]_{i\in V}^{\top}\\
-\left[\left(\text{proj}\left(\aleph_{i}\right)\right)^{\top}\right]_{i\in V}^{\top}
\end{array}\right].\label{eq:closed-loop-dynamics}
\end{equation}
}Consider the candidate Lyapunov function
\begin{equation}
V\left(z\right)\triangleq\frac{1}{2}z^{\top}Pz,\label{eq:lyapunov}
\end{equation}
where $P\triangleq\text{blkdiag}\{I_{n},I_{nN},\Gamma^{-1}\}\in\mathbb{R}^{\varphi\times\varphi}$
and $V:\mathbb{R}^{\varphi}\to\mathbb{R}_{\geq0}$. By the Rayleigh
quotient, (\ref{eq:lyapunov}) satisfies
\begin{equation}
\lambda_{1}\left\Vert z\right\Vert ^{2}\leq V\left(z\right)\leq\lambda_{2}\left\Vert z\right\Vert ^{2},\label{eq:rayleigh-ritz}
\end{equation}
where $\lambda_{1}\triangleq\frac{1}{2}\min\{1,\lambda_{\text{min}}(\Gamma^{-1})\}$
and $\lambda_{2}\triangleq\frac{1}{2}\max\{1,\lambda_{\text{max}}(\Gamma^{-1})\}$.
Based on the subsequent proposition, define $\lambda_{3}\in\mathbb{R}_{>0}$
as $\lambda_{3}\triangleq\min\left\{ \frac{k_{1}}{2}\overline{g}N-\overline{g}N-\frac{1}{2}N^{3}-\frac{1}{2}N-\frac{\epsilon_{1}}{2},\right.\frac{k_{2}}{4}-k_{1}\overline{g}-\frac{1}{2}k_{1}^{2}\overline{g}^{2}-\frac{1}{2}\overline{g}^{2}-\frac{1}{2}k_{1}^{4}\overline{g}^{2}\left.,\frac{k_{3}}{2}\right\} $.
\begin{proposition}
\label{proposition:gain_conditions} If the user-defined gains $\epsilon_{1}\in\mathbb{R}$,
$k_{1}\in\mathbb{R}$, $k_{2}\in\mathbb{R}$, and $k_{3}\in\mathbb{R}$
are sequentially selected to satisfy the sufficient gain conditions
$\epsilon_{1}>0$, $k_{1}>2+\frac{1}{\overline{g}N}\left(N^{3}+N+\epsilon_{1}\right)$,
$k_{2}>4k_{1}\overline{g}+2\overline{g}^{2}\left(k_{1}^{4}+k_{1}^{2}+1\right)$,
and $k_{3}>0$, respectively, then $\lambda_{3}>0$.
\end{proposition}
The universal function approximation property of GNNs only holds over
a compact domain; therefore, we require that the inputs to $\Phi$
must lie on a compact domain for all $t\in\mathbb{R}_{\geq0}$. We
enforce this condition by proving that $R_{i}\in\mathcal{Y}_{i}$,
for all $i\in V$. Let $\rho:\mathbb{R}_{\geq0}\to\mathbb{R}_{\geq0}$
denote a strictly increasing function and define $\overline{\rho}(\cdot)\triangleq\rho(\cdot)-\rho(0)$,
where $\overline{\rho}$ is strictly increasing and invertible. Consider
the compact domain $\mathcal{D}\subset\mathbb{R}^{\varphi}$, defined
as{\small
\begin{equation}
\mathcal{D}\triangleq\left\{ \psi\in\mathbb{R}^{\varphi}:\left\Vert \psi\right\Vert \leq\overline{\rho}^{-1}\left(k_{2}\left(\lambda_{3}-\lambda_{4}\right)-\rho(0)\right)\right\} ,\label{eq:state-space}
\end{equation}
}where $\lambda_{4}\in\mathbb{R}_{>0}$ is a user-defined rate of
convergence. Define the set $\Omega\triangleq G\times\mathcal{Y}$,
where $\mathcal{Y}\triangleq\{\psi\in\mathbb{R}^{Nd^{(in)}}:\psi_{i}\in\mathcal{Y}_{i},\ \forall\ensuremath{i\in V}\}$,
and $\mathcal{Y}_{i}$ is defined as
\begin{equation}
\begin{gathered}\mathcal{Y}_{i}\triangleq\left\{ \psi\right.\in\mathbb{R}^{d^{(in)}}:\left\Vert \psi\right\Vert \leq\left(1+N\left(k_{1}+1\right)\right)\\
\cdot\overline{\rho}^{-1}\left(k_{2}\left(\lambda_{3}-\lambda_{4}\right)-\rho(0)\right)+\left(1+N\right)\overline{x}_{d}\Bigr\},
\end{gathered}
\label{eq:compact-domain}
\end{equation}
For the dynamics described by (\ref{eq:closed-loop-dynamics}), the
set of initial conditions $\mathcal{S}\subset\mathcal{D}$ is defined
as{\scriptsize
\begin{align}
\mathcal{S} & \triangleq\left\{ \psi\in\mathbb{R}^{\varphi}:\left\Vert \psi\right\Vert <\sqrt{\frac{\lambda_{1}}{\lambda_{2}}}\overline{\rho}^{-1}\left(k_{2}\left(\lambda_{3}-\lambda_{4}\right)-\rho(0)\right)-\sqrt{\frac{\upsilon}{\lambda_{4}}}\right\} ,\label{eq:stabilizing-initial-conditions}
\end{align}
}and the uniform ultimate bound is defined as
\begin{equation}
\mathcal{U}\triangleq\left\{ \psi\in\mathbb{R}^{\varphi}:\left\Vert \psi\right\Vert \leq\sqrt{\frac{\lambda_{2}\upsilon}{\lambda_{1}\lambda_{4}}}\right\} ,\label{eq:ultimate-bound-ball}
\end{equation}
where $\upsilon\in\mathbb{R}_{>0}$ is defined as $\upsilon\triangleq\frac{\overline{\varepsilon}^{2}}{k_{2}}+\frac{\overline{\dot{x}}_{d}^{2}}{2\overline{g}N}+\frac{\overline{h}^{2}}{2\epsilon_{1}}+\frac{1}{2}\left(2N+1\right)^{2}k_{3}\overline{\theta}^{2}N$.
\begin{theorem}
\label{theorem:lyapunov-stability} For the influencing nodes' dynamics
given by (\ref{eq:influencer-dynamics}) and target node's dynamics
given by (\ref{eq:target-dynamics}), the controller and adaptive
update law in (\ref{eq:controller}) and (\ref{eq:update-law}), respectively,
guarantee that for any $z(t_{0})\in\mathcal{S}$, $z$ exponentially
converges to $\mathcal{U}$ in the sense that{\small
\[
\left\Vert z(t)\right\Vert \leq\sqrt{\frac{\lambda_{2}}{\lambda_{1}}\left(\frac{\upsilon}{\lambda_{4}}+e^{-\frac{\lambda_{4}}{\lambda_{2}}\left(t-t_{0}\right)}\left(\left\Vert z\left(t_{0}\right)\right\Vert ^{2}-\frac{\upsilon}{\lambda_{4}}\right)\right)},
\]
}for all $t\in[t_{0},\infty)$, provided that Proposition \ref{proposition:gain_conditions}
is satisfied, the sufficient gain condition $\lambda_{3}>\lambda_{4}+\frac{1}{k_{2}}\rho\left(\sqrt{\frac{\lambda_{2}\upsilon}{\lambda_{1}\lambda_{4}}}\right)$
is satisfied, and Assumptions \ref{ass:interaction-function-bounds}-\ref{ass:convex-loss-function}
hold.
\end{theorem}
\begin{proof}
Taking the time derivative of (\ref{eq:lyapunov}) and substituting
(\ref{eq:closed-loop-dynamics}) into the resulting expression yields{\small
\begin{equation}
\begin{aligned}\dot{V} & =-\left(k_{1}-1\right)\sum_{i\in V}g\left(x_{0},y_{i}\right)e^{\top}e+e^{\top}h\left(x_{0}\right)\\
 & -k_{2}\eta^{\top}\eta-e^{\top}\dot{x}_{d}+\sum_{i\in V}g\left(x_{0},y_{i}\right)e^{\top}\eta_{i}\\
 & +k_{1}\sum_{i\in V}\sum_{j\in V\backslash\mathcal{\overline{N}}_{i}}g\left(x_{0},y_{j}\right)\eta_{i}^{\top}\eta_{j}\\
 & +k_{1}\left(1-k_{1}\right)\sum_{i\in V}\sum_{j\in V\backslash\mathcal{\overline{N}}_{i}}g\left(x_{0},y_{j}\right)\eta_{i}^{\top}e+\eta^{\top}\varepsilon\\
 & +\sum_{i\in V}\sum_{j\in\mathcal{\mathcal{\overline{N}}}_{i}^{k-1}}\eta_{i}^{\top}\frac{\partial\widehat{\phi}_{i}}{\partial\widehat{\theta}_{j}}\widetilde{\theta}_{i}+\eta^{\top}\chi\\
 & -\sum_{i\in V}\widetilde{\theta}_{i}^{\top}\Gamma_{i}^{-1}\text{proj}\left(\aleph_{i}\right),
\end{aligned}
\label{eq:lyapunov-derivative-1}
\end{equation}
}where $\chi\triangleq[\chi_{i}^{\top}]_{i\in V}^{\top}\in\mathbb{R}^{nN}$.
From \cite[Appendix E.4]{Krstic.Kanellakopoulos.ea1995}, we have
that $\text{proj}\left(\aleph_{i}\right)\leq\aleph_{i}$. Using this
fact, performing algebraic manipulation, and applying (\ref{eq:parameter-estimation-error})
yields{\small
\begin{equation}
\begin{aligned}-\sum_{i\in V}\widetilde{\theta}_{i}^{\top}\Gamma_{i}^{-1}\text{proj}\left(\aleph_{i}\right) & \leq-\sum_{i\in V}\sum_{j\in\mathcal{\overline{N}}_{i}^{k-1}}\left(\widetilde{\theta}_{i}^{\top}\frac{\partial\widehat{\phi}_{i}}{\partial\widehat{\theta}_{j}}^{\top}\eta_{i}\right)\\
 & +k_{3}\sum_{i\in V}\widetilde{\theta}_{i}^{\top}\widehat{\theta}_{i}\\
 & +k_{3}\sum_{i\in V}\sum_{j\in\mathcal{N}_{i}}\widetilde{\theta}_{i}^{\top}\left(\widetilde{\theta}_{j}-\widetilde{\theta}_{i}\right)\\
 & +k_{3}\sum_{i\in V}\sum_{j\in\mathcal{N}_{i}}\widetilde{\theta}_{i}^{\top}\left(\theta_{i}^{\ast}-\theta_{j}^{\ast}\right).
\end{aligned}
\label{eq:update-law-bounding-1}
\end{equation}
}Using (\ref{eq:theta_star}), the definition of the graph Laplacian,
and expanding $\widehat{\theta}_{i}$ according to (\ref{eq:parameter-estimation-error}),
(\ref{eq:update-law-bounding-1}) is upper bounded as

{\small
\begin{equation}
\begin{aligned}-\sum_{i\in V}\widetilde{\theta}_{i}^{\top}\Gamma_{i}^{-1}\text{proj}\left(\aleph_{i}\right) & \leq-\sum_{i\in V}\sum_{j\in\mathcal{\overline{N}}_{i}^{k-1}}\left(\widetilde{\theta}_{i}^{\top}\frac{\partial\widehat{\phi}_{i}}{\partial\widehat{\theta}_{j}}^{\top}\eta_{i}\right)\\
 & +k_{3}\sum_{i\in V}\widetilde{\theta}_{i}^{\top}\overline{\theta}-k_{3}\widetilde{\theta}^{\top}\widetilde{\theta}\\
 & +k_{3}\sum_{i\in V}\sum_{j\in\mathcal{N}_{i}}\widetilde{\theta}_{i}^{\top}\left(\theta_{i}^{\ast}-\theta_{j}^{\ast}\right).
\end{aligned}
\label{eq:update-law-bounding-2}
\end{equation}
}Using (\ref{eq:theta_star}), substituting (\ref{eq:update-law-bounding-2})
into (\ref{eq:lyapunov-derivative-1}), upper bounding all terms by
their norms, and applying the Cauchy-Schwarz inequality yields
\begin{equation}
\begin{aligned}\dot{V} & \leq-\left(k_{1}-1\right)\overline{g}N\left\Vert e\right\Vert ^{2}-\left(k_{2}-k_{1}\overline{g}\right)\left\Vert \eta\right\Vert ^{2}\\
 & -k_{3}\left\Vert \widetilde{\theta}\right\Vert ^{2}+\left\Vert e\right\Vert \overline{\dot{x}}_{d}+\overline{g}\sqrt{N}\left\Vert \eta\right\Vert \left\Vert e\right\Vert \\
 & +k_{1}\left(k_{1}+1\right)\overline{g}N\sqrt{N}\left\Vert \eta\right\Vert \left\Vert e\right\Vert \\
 & +\left\Vert e\right\Vert \overline{h}+\left\Vert \eta\right\Vert \left\Vert \chi\right\Vert +\left\Vert \eta\right\Vert \overline{\varepsilon}\\
 & +\left(2N+1\right)k_{3}\overline{\theta}\sqrt{N}\left\Vert \widetilde{\theta}\right\Vert .
\end{aligned}
\label{eq:lyapunov-derivative-2}
\end{equation}
By Lemma \ref{lemma:lagrange-bounds} and the use of a bounded search
space for $\theta^{\ast}$ in (\ref{eq:theta_star}), there exist
some $\rho_{G_{1},j}:\mathbb{R}_{\geq0}\to\mathbb{R}_{\geq0}$ and
$\rho_{T_{1},j}:\mathbb{R}_{\geq0}\to\mathbb{R}_{\geq0}$ such that
$\rho_{G_{1},j}$ and $\rho_{T_{1},j}$ are strictly increasing functions,
where $\|\chi_{i}\|\leq2\overline{\theta}\sum_{j\in\mathcal{\mathcal{N}}_{i}^{k-1}}\rho_{G_{1},j}(\|R\|)+\sum_{j\in\mathcal{\mathcal{\overline{N}}}_{i}^{k-1}}\rho_{T_{1},j}(\|R\|)\|\widetilde{\theta}_{j}\|^{2}$,
for all $i\in V$. From (\ref{eq:theta_star}) and the definition
of $R$, there exist some $\rho_{1,i}:\mathbb{R}_{\geq0}\to\mathbb{R}_{\geq0}$
where $\rho_{1,i}$ are strictly increasing functions such that $\|\chi_{i}\|\leq\rho_{1,i}(\|z\|)\|z\|$,
for all $i\in V$. Using the Cauchy-Schwarz inequality, it follows
that $\|\eta\|\|\chi\|\leq\|\eta\|\rho_{2}(\|z\|)\|z\|$, where $\rho_{2}:\mathbb{R}_{\geq0}\to\mathbb{R}_{\geq0}$
is a strictly increasing function. By completing the square, we upper
bound $\lVert\eta\rVert\lVert\chi\rVert$ as
\begin{equation}
\left\Vert \eta\right\Vert \left\Vert \chi\right\Vert \leq\frac{1}{k_{2}}\rho\left(\left\Vert z\right\Vert \right)\left\Vert z\right\Vert ^{2}+\frac{k_{2}}{2}\left\Vert \eta\right\Vert ^{2},\label{eq:nonlinear-damping}
\end{equation}
where $\rho:\mathbb{R}_{\geq0}\to\mathbb{R}_{\geq0}$ is a strictly
increasing function that satisfies $\frac{1}{2}\rho_{2}^{2}\left(\left\Vert z\right\Vert \right)\leq\rho\left(\left\Vert z\right\Vert \right)$.
Substituting (\ref{eq:nonlinear-damping}) into (\ref{eq:lyapunov-derivative-2}),
applying Assumptions \ref{ass:interaction-function-bounds} and \ref{ass:goal-location-bounds},
and applying Young's inequality yields{\small
\begin{equation}
\begin{aligned}\dot{V} & \leq-\left(k_{1}\overline{g}N-\overline{g}N-\frac{1}{2}N^{3}-\frac{1}{2}N\right)\left\Vert e\right\Vert ^{2}\\
 & -\left(\frac{k_{2}}{2}-k_{1}\overline{g}-\frac{1}{2}k_{1}^{2}\overline{g}^{2}-\frac{1}{2}\overline{g}^{2}-\frac{1}{2}k_{1}^{4}\overline{g}^{2}\right)\left\Vert \eta\right\Vert ^{2}\\
 & -k_{3}\left\Vert \widetilde{\theta}\right\Vert ^{2}+\frac{1}{k_{2}}\rho\left(\left\Vert z\right\Vert \right)\left\Vert z\right\Vert ^{2}+\left\Vert e\right\Vert \overline{\dot{x}}_{d}\\
 & +\left\Vert \eta\right\Vert \overline{\varepsilon}+\left(2N+1\right)k_{3}\overline{\theta}\sqrt{N}\left\Vert \widetilde{\theta}\right\Vert +\left\Vert e\right\Vert \overline{h}.
\end{aligned}
\label{eq:lyapunov-derivative-3}
\end{equation}
}Completing the square and applying the triangle inequality to (\ref{eq:lyapunov-derivative-3})
yields
\begin{equation}
\begin{aligned}\dot{V} & \leq-\left(\frac{k_{1}}{2}\overline{g}N-\overline{g}N-\frac{1}{2}N^{3}-\frac{1}{2}N-\frac{\epsilon_{1}}{2}\right)\left\Vert e\right\Vert ^{2}\\
 & -\left(\frac{k_{2}}{4}-k_{1}\overline{g}-\frac{1}{2}k_{1}^{2}\overline{g}^{2}-\frac{1}{2}\overline{g}^{2}-\frac{1}{2}k_{1}^{4}\overline{g}^{2}\right)\left\Vert \eta\right\Vert ^{2}\\
 & -\frac{k_{3}}{2}\left\Vert \widetilde{\theta}\right\Vert ^{2}+\Biggl(\frac{\overline{\varepsilon}^{2}}{k_{2}}+\frac{\overline{h}^{2}}{2\epsilon_{1}}+\frac{\overline{\dot{x}}_{d}^{2}}{2\overline{g}N}\\
 & +\frac{1}{2}\left(2N+1\right)^{2}k_{3}\overline{\theta}^{2}N\Biggr)+\frac{1}{k_{2}}\rho\left(\left\Vert z\right\Vert \right)\left\Vert z\right\Vert ^{2}.
\end{aligned}
\label{eq:lyapunov-derivative-4}
\end{equation}
Substituting the definitions of $\lambda_{3}$ and $\upsilon$ into
(\ref{eq:lyapunov-derivative-4}) gives
\begin{equation}
\dot{V}\leq-\text{\ensuremath{\left(\lambda_{3}-\frac{1}{k_{2}}\ensuremath{\rho\left(\left\Vert z\right\Vert \right)}\right)}}\left\Vert z\right\Vert ^{2}+\upsilon,\label{eq:stability-4}
\end{equation}
for all $t\in\mathcal{I}$. Since the solution $t\mapsto z$$(t)$
is continuous, there exists a time interval $\mathcal{I}\triangleq[t_{0},t_{1}]$
with $t_{1}>t_{0}$ such that $z$$\in\mathcal{D}$ for all $t\in\mathcal{I}$.
Since $z(t_{0})\in\mathcal{S}$ and $\rho$ is a strictly increasing
function, (\ref{eq:stability-4}) can be written as
\begin{equation}
\dot{V}\leq-\frac{\lambda_{4}}{\lambda_{2}}V\left(z\right)+\upsilon,\label{eq:lyapunov-derivative-5}
\end{equation}
for all $t\in\mathcal{I}$. Solving the differential inequality in
(\ref{eq:lyapunov-derivative-5}) over $\mathcal{I}$ yields{\footnotesize
\begin{align}
V\left(z\left(t\right)\right) & \leq V\left(z\left(t_{0}\right)\right)e^{-\frac{\lambda_{4}}{\lambda_{2}}\left(t-t_{0}\right)}+\frac{\lambda_{2}\upsilon}{\lambda_{4}}\left(1-e^{-\frac{\lambda_{4}}{\lambda_{2}}\left(t-t_{0}\right)}\right).\label{eq:lyapunov-solution}
\end{align}
}Applying (\ref{eq:rayleigh-ritz}) to (\ref{eq:lyapunov-solution})
yields{\small
\begin{equation}
\left\Vert z\left(t\right)\right\Vert \leq\sqrt{e^{-\frac{\lambda_{4}}{\lambda_{2}}\left(t-t_{0}\right)}\left(\frac{\lambda_{2}}{\lambda_{1}}\left\Vert z\left(t_{0}\right)\right\Vert ^{2}-\frac{\lambda_{2}\upsilon}{\lambda_{1}\lambda_{4}}\right)+\frac{\lambda_{2}\upsilon}{\lambda_{1}\lambda_{4}}},\label{eq:state-traj-final-sol}
\end{equation}
}for all $t\in\mathcal{I}$. Next, we must show that $\mathcal{I}$
can be extended to $[t_{0},\infty).$ Let $t\mapsto z(t)$ be a solution
to the ordinary differential equation in (\ref{eq:closed-loop-dynamics})
with initial condition $z(t_{0})\in\mathcal{S}$, as defined in (\ref{eq:stabilizing-initial-conditions}).
By \cite[Lemma E.1]{Krstic.Kanellakopoulos.ea1995}, the projection
operator is locally Lipschitz in its arguments. Then, the right hand
side of (\ref{eq:closed-loop-dynamics}) is piecewise continuous in
$t$ and locally Lipschitz in $z$ for $t\geq t_{0}$ where $z\in\mathbb{R}^{\varphi}$.
By (\ref{eq:state-traj-final-sol}), it follows that
\begin{equation}
\left\Vert z\left(t\right)\right\Vert <\sqrt{\frac{\lambda_{1}}{\lambda_{2}}}\left\Vert z\left(t_{0}\right)\right\Vert +\sqrt{\frac{\lambda_{2}\upsilon}{\lambda_{1}\lambda_{4}}},\label{eq:size-of-D}
\end{equation}
for all $t\in\mathcal{I}$. Since $z(t_{0})\in\mathcal{S}$ for all
$t\in\mathcal{I}$, we have
\begin{equation}
\left\Vert z\left(t_{0}\right)\right\Vert <\sqrt{\frac{\lambda_{1}}{\lambda_{2}}}\overline{\rho}^{-1}\left(k_{2}\left(\lambda_{3}-\lambda_{4}\right)-\rho(0)\right)+\sqrt{\frac{\upsilon}{\lambda_{4}}}.\label{eq:size-of-S}
\end{equation}
Applying (\ref{eq:size-of-S}) to (\ref{eq:size-of-D}) gives $\|z(t_{0})\|\leq\overline{\rho}^{-1}\left(k_{2}\left(\lambda_{3}-\lambda_{4}\right)-\rho(0)\right)$
for all $t\in\mathcal{I}$, which by the definition of (\ref{eq:state-space}),
yields $z(t)\in\mathcal{D}$ for all $t\in\mathcal{I}$. Since $z(t)$
remains in the compact set $\mathcal{D}\subset\mathbb{R}^{\varphi}$
for all $t\in\mathcal{I}$, a unique solution $z(t)$ exists for all
$t\geq t_{0}$ \cite[Theorem 3.3]{Khalil2002}. Therefore, $\mathcal{I}=[t_{0},\infty)$.
Thus, for $z(t_{0})\in\mathcal{S}$, (\ref{eq:state-traj-final-sol})
holds for all $z(t)\in\mathcal{D}$ for all $t\in\mathcal{I}$. The
limit of (\ref{eq:state-traj-final-sol}) as $t\to\infty$ yields
$\|z\|\leq\sqrt{\frac{\lambda_{2}\upsilon}{\lambda_{1}\lambda_{4}}}$,
which indicates that $z(t)$ converges to the set $\mathcal{U}$,
as defined in (\ref{eq:ultimate-bound-ball}).

We must prove that $R_{i}\in\mathcal{Y}_{i}$ for all $i\in V$. Using
(\ref{eq:tracking-error}), Assumption \ref{ass:goal-location-bounds},
and the triangle inequality yields $\|x_{0}\|\leq\overline{x}_{d}+\|e\|$.
Similarly, using (\ref{eq:backstepping-error}), (\ref{eq:influence-strategy}),
and the triangle inequality yields $\left\Vert y_{i}\right\Vert \leq k_{1}\left\Vert e\right\Vert +\overline{x}_{d}+\left\Vert \eta\right\Vert $,
for all $i\in V$. By the definitions of $Q_{i}$ and $z$, it follows
that
\begin{equation}
\left\Vert R_{i}\right\Vert \leq\left(1+N\left(k_{1}+1\right)\right)\left\Vert z\right\Vert +\left(1+N\right)\overline{x}_{d},\label{eq:compact-domain-size}
\end{equation}
for all $i\in V$. Therefore, $G\times R\in\Omega$ for all $t\in\mathcal{I}$
and the universal function approximation property of GNNs described
in \cite[Lemma 3]{arxivFallin.Nino.ea2025} holds for all $t\in\mathcal{I}$.

Since $z\in\mathcal{D}$ for all $t\in[t_{0},\infty)$ it follows
that from (\ref{eq:compact-domain}) and (\ref{eq:compact-domain-size})
that $R_{i}\in\mathcal{Y}_{i}$ for all $i\in V$. Therefore, $z\in D$
implies $G\times R\in\Omega$, and the universal function approximation
property for GNNs holds everywhere.

Since $\|z\|\leq\overline{\rho}^{-1}\left(k_{2}\left(\lambda_{3}-\lambda_{4}\right)-\rho(0)\right)$,
$\|e\|,\|\eta\|,\|\widetilde{\theta}\|$ are bounded. Therefore, $e,\eta,\widetilde{\theta}\in\mathcal{L}_{\infty}$.
By Assumption \ref{ass:goal-location-bounds}, $x_{d}\in\mathcal{L}_{\infty}$
and since $e,\eta\in\mathcal{L}_{\infty}$, then $R\in\mathcal{L}_{\infty}$.
Based on the fact that $\widetilde{\theta}\in\mathcal{L}_{\infty}$,
the use of the projection operator in (\ref{eq:update-law}), and
the use of a bounded search space in (\ref{eq:theta_star}), we have
$\widehat{\theta}\in\mathcal{L}_{\infty}$ where $\widehat{\theta}\triangleq[\widehat{\theta}_{i}^{\top}]_{i\in V}^{\top}\in\mathbb{R}^{pN}$.
Let $\widehat{\Phi}\triangleq[\widehat{\phi}_{i}^{\top}]^{\top}\in\mathbb{R}^{nN}$.
Based on the use of activation functions with bounded derivatives
and the fact that $R,\widehat{\theta}\in\mathcal{L}_{\infty}$, the
norm of the derivative of the GNN $\hat{\Phi}$ with respect to the
weight estimates and the norm of the GNN $\hat{\Phi}$ are bounded.
Let $u\triangleq[u_{i}^{\top}]_{i\in V}^{\top}\in\mathbb{R}^{nN}$.
By Assumption \ref{ass:goal-location-bounds} and since $\eta,\widehat{\theta}\in\mathcal{L}_{\infty}$,
the norm of derivative of the GNN $\hat{\Phi}$ with respect to the
weight estimates is bounded, and the norm of the GNN $\hat{\Phi}$
is bounded, $u\in\mathcal{L}_{\infty}$. Since $\eta,\widehat{\theta}\in\mathcal{L}_{\infty}$
and that the norm of derivative of the GNN $\hat{\Phi}$ with respect
to the weight estimates is bounded, $\dot{\widehat{\theta}}\in\mathcal{L}_{\infty}$.
\end{proof}

\section{Simulation}

Using the control law in (\ref{eq:controller}) and the adaptive weight
update law in (\ref{eq:update-law}), a simulation is performed where
a network of $N=4$ nodes with a fully connected communication graph
indirectly influences a target node to follow a desired trajectory,
given by
\[
x_{d}=\left[\begin{array}{c}
10\sin\left(0.01t\right)\\
10\sin\left(0.025t\right)\cos\left(0.025t\right)\\
5\sin\left(0.075t\right)
\end{array}\right]\ m.
\]
The target node's interaction dynamics were selected as $g(x_{0},y_{i})=0.1\ \text{exp}(-\frac{1}{1,000,000}(x_{0}-y_{i})^{\top}(x_{0}-y_{i}))\ s^{-1}$.
\cite{Licitra.Bell.ea2019}. The inter-agent interaction dynamics
were selected as $f(Q_{i})=50\sum_{j\in\mathcal{N}_{i}}\frac{y_{i}-y_{j}}{\|y_{i}-y_{j}\|^{3}}\ m/s$
\cite{Sebastian.Montijano.ea2022}. The target drift dynamics were
selected such that $h(x_{0})=[-0.057\cos(x_{0,1}),0.03\sin(x_{0,2}),-0.008\cos(x_{0,3})]^{\top}\ m/s$.
These dynamics enable a repulsive effect between nodes. The initial
positions of the influencing nodes were selected as $y_{1}(t_{0})=[\begin{array}{ccc}
-6 & -1 & 8\end{array}]^{\top}\ m$, $y_{2}(t_{0})=[\begin{array}{ccc}
6 & 4 & -2\end{array}]^{\top}\ m$, $y_{3}(t_{0})=[\begin{array}{ccc}
4 & -6 & 1\end{array}]^{\top}\ m$, and $y_{4}(t_{0})=[\begin{array}{ccc}
-4, & 6, & -2\end{array}]^{\top}\ m$. The target node is initialized at $x_{0}(t_{0})=[\begin{array}{ccc}
6 & -4 & 2\end{array}]^{\top}\ m$. For all $j\in\{0,\ldots,k\}$, the Lb-GNN weight estimates were
initialized with random values drawn from a uniform distribution $W_{i}^{(j)}\sim U(0,0.3)$,
for all $i\in V$. The size of the parameter search space was selected
as $\overline{\theta}=10$. Due to the message-passing framework of
the GNN, the number of required successive communications between
neighbors is equivalent to the number of message-passing layers in
the network. In the multi-agent control literature, the number of
GNN layers is typically limited to between 1 and 4 layers\cite{Gama.Tolstaya.ea2021,Li.Gama.ea2020}.
The Lb-GNN $\Phi$ has two hidden layers ($k=2$) and eight neurons
per hidden layer, for all $i\in V$. The swish activation function
was used for all hidden layers, and $\text{\text{tanh}(\ensuremath{\cdot})}$
was used on the output layer. For all $i\in V$, the gain values were
empirically selected as $k_{1}=3.5$, $k_{2}=12$, $k_{3}=0.001$,
and $\Gamma_{i}=2\cdot I_{p\times p}$. The simulation was run for
360 seconds. A three-dimensional visualization of the trajectories
of the influencing and target nodes is presented in Figure \ref{fig:trajectories}.

\begin{figure}
\centering{}\includegraphics[width=1\columnwidth]{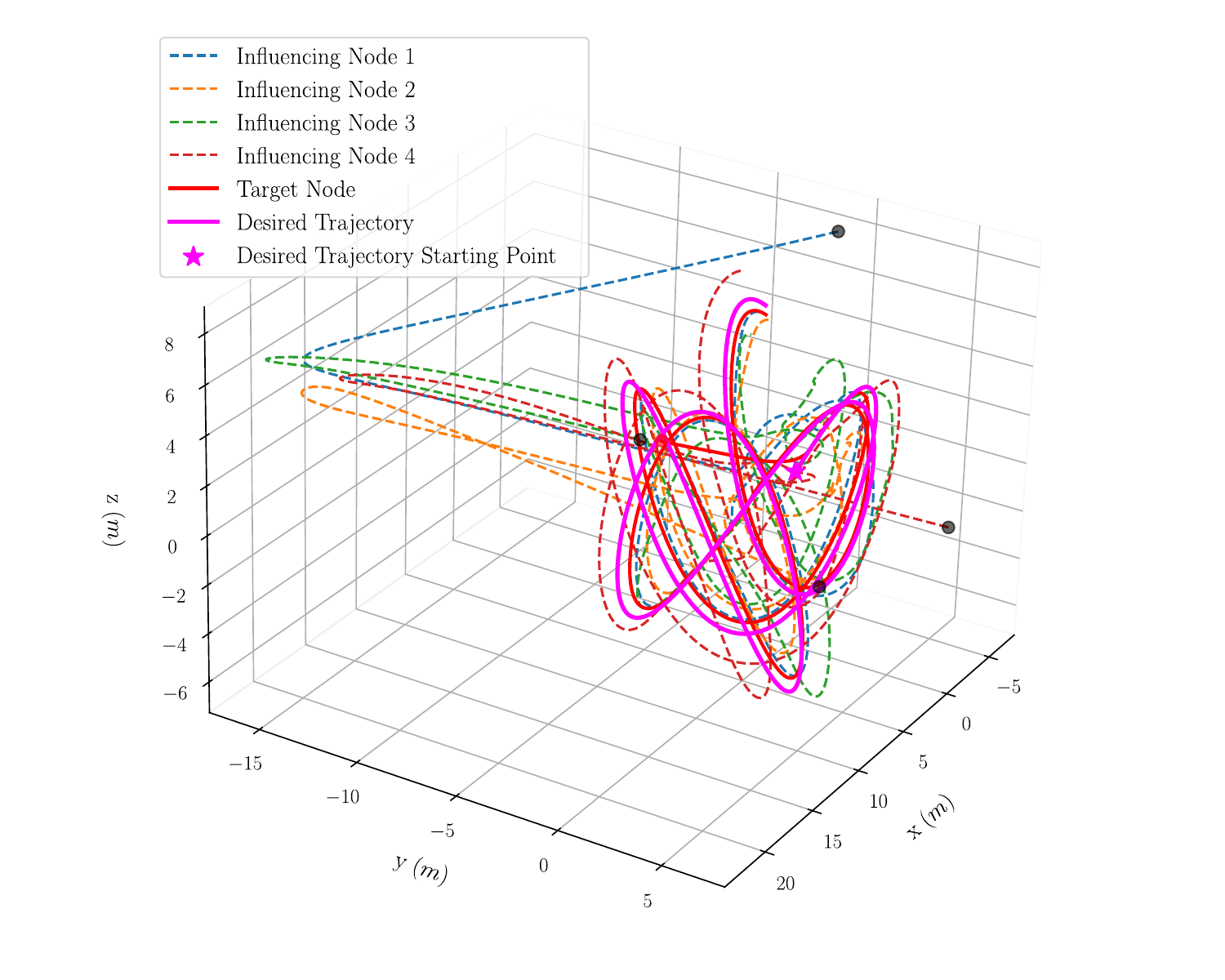}

\caption{\label{fig:trajectories} Visualization of the trajectories of the
target node and the influencing nodes for $t\in[0,360]\ s$.}
\end{figure}
The simulation resulted in an RMS position tracking error of $\lVert e\rVert_{\text{RMS}}=0.61\ m$,
a mean RMS control effort of $\lVert u\rVert_{\text{RMS}}=19.57\ m/s$,
and a mean RMS function approximation error of $\lVert\widetilde{\Phi}\rVert_{\text{RMS}}=15.95\ m/s$,
where $\widetilde{\Phi}\triangleq\Phi(R,\theta)-H(R)$ denotes the
function approximation error for the GNN $\Phi$.
\begin{figure}
\centering{}\includegraphics[width=1\columnwidth]{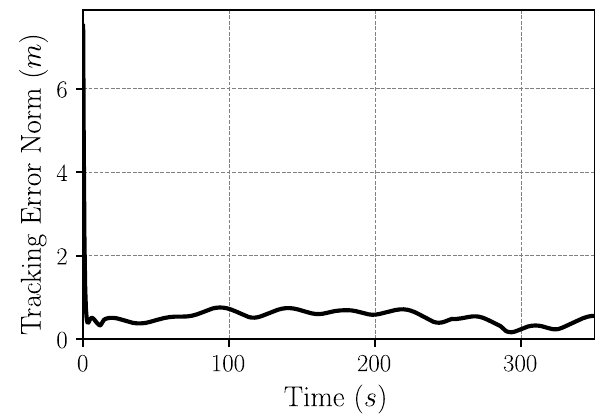}

\caption{\label{fig:tracking_error} Visualization of the position tracking
error of the target node for $t\in[0,360]\ s$.}
\end{figure}

Figure \ref{fig:tracking_error} shows the ability of the designed
control law to allow multiple nodes to influence a target node to
within a neighborhood of a desired trajectory. As seen in Figure \ref{fig:tracking_error},
the norm of the tracking error converges to within an ultimate bound
of approximately $1\ m$ after 5 seconds.

\section{Conclusion}

In this work, we developed the first Lb-GNN-based control law for
the indirect influence problem. Our method enables a team of cooperative
nodes to regulate a target node to within a neighborhood of a desired
trajectory despite uncertain inter-agent interaction dynamics. Numerical
simulations validate the theoretical findings. Future work will examine
influence strategies with formation control in the presence of multiple
targets.

\bibliographystyle{IEEEtran}
\bibliography{bib/encr,bib/master,bib/ncr}

\end{document}